\documentclass[11pt]{article}

\usepackage[letterpaper,margin=1.12in]{geometry}
\usepackage{newtxtext,newtxmath}
\usepackage{microtype}
\usepackage{amsmath,booktabs,array}

\usepackage{amsthm}
\usepackage[hidelinks]{hyperref}
\usepackage{enumitem}

\setlist{nosep,leftmargin=*}
\newtheorem{theorem}{Theorem}
\newtheorem{proposition}[theorem]{Proposition}
\newtheorem{lemma}[theorem]{Lemma}
\newtheorem{corollary}[theorem]{Corollary}
\newtheorem{remark}[theorem]{Remark}
\theoremstyle{definition}

\newcommand{\R}{\mathbb{R}}

\newcommand{\succeqB}{\succeq_{\!B}}
\newcommand{\MGF}{\mathsf M}
\newcommand{\KL}{\operatorname{KL}}

\title{\vspace{-1.5em}\textbf{Blackwell Boundaries} }
\author{Shuo Li Liu\\\texttt{shuoliliu@princeton.edu}}
\date{September 2026}

\begin{document}
\maketitle
\vspace{-1.2em}

\begin{abstract}
We give a finite simplex theorem. When the number of states equals the
number of source signals and the source posterior likelihood-ratio vectors
form a simplex, Blackwell dominance is characterized by a universal
barycentric-coordinate condition. We also give an explicit countable-state,
countable-signal diagnostic-monitor theorem with a closed-form universal
criterion and a closed-form garbling; tensorization yields dominance at every
sample size. Finally, we provide a three-state, two-signal counterexample to
the sufficiency conjectured by Mu, Pomatto, Strack, and Tamuz for their
many-state moment-generating-function conditions for large-sample Blackwell
dominance. The pair satisfies strict comparisons on both parameter domains,
all ordered Kullback--Leibler inequalities, bounded likelihood ratios and
pairwise genericity, but dominance fails at every positive sample size. The
finite theorem, obstruction, and continuum inequalities are verified in Lean
4.30 with Mathlib; the countable diagnostic proof is given analytically.
\end{abstract}


\section{Introduction}

Blackwell dominance is the standard measure of informational superiority. If a stochastic kernel $K$ maps source signals into target signals,
\begin{equation}
Q_i(y)=\sum_x P_i(x)K(x,y)\qquad\text{for every state }i,
\label{eq:garbling}
\end{equation}
then every decision rule available under $Q$ can be reproduced under $P$. In applications, a richer monitoring log can be deliberately degraded, whereas a poorer log need not reproduce every policy based on the richer log.

The product problem with several states is subtle. Product likelihood and divergence inequalities may hold even when no finite Blackwell garbling exists. We isolate the local mechanism responsible for the failure. A target atom lies on a likelihood-ratio equality face. Equality of the product of ratios then forces equality in each coordinate. The required source atom is absent from the source support. In tropical coordinates, this appears as equality at a nonvertex boundary point.

Our main positive results are the finite simplex theorem and the countable
diagnostic-monitor theorem stated below.
The main negative result is a three-state, two-signal counterexample that
satisfies Mu's full necessary-condition package but fails Blackwell
dominance at every sample size. We place the audit problem, the residual-coupling
lemma, and the analytic details in the appendices.

Restricted noisy-audit, common-noise, and modular completions appear
in the appendices. They serve as benchmark model classes rather than assumptions in the
unrestricted results.

The counterexample argument is self-contained. The accompanying Lean source
records its arithmetic, analytic inequalities, and finite transport arguments.

\subsection*{Relation to the broader conjecture}

The displayed pair satisfies the full set of necessary conditions proposed for
many-state dominance in large samples, including the continuum positive and
negative MGF inequalities, all ordered KL inequalities, boundedness, and
pairwise genericity. Yet the all-$n$ obstruction below rules out
eventual dominance. One valid three-state, two-signal experiment is therefore
enough to refute the conjecture as a universal claim over finite experiments.
We do not claim that every experiment with more than three states, or every
signal cardinality, fails; such stronger statements require separate
constructions.

Any primitive condition $(\mathrm{C4})$ that is jointly necessary and
sufficient with Mu's conditions for the unrestricted finite class must fail
for the pair in \eqref{eq:experiments}. A successful non-circular repair must
therefore exclude the nonrectangular equality-face geometry exhibited below;
it cannot follow from boundedness, positivity, genericity, or the MGF/KL
inequalities alone.
The example also satisfies the usual monotone-likelihood-ratio/TP2 property
after relabelling the states and ordering the two signals suitably, and its
probability matrices have full column rank. Those standard regularity
conditions therefore do not repair the many-state conjecture either.

\section{Experiments and the product family}

Let the state space be $I=\{0,1,2\}$ and the signal space be finite. An experiment is a family $E=(E_i)_{i\in I}$ of probability distributions on signals. We write $P\succeqB Q$ when there is a row-stochastic kernel satisfying \eqref{eq:garbling}. For $n$ conditionally independent observations, the product experiments are $P^{\otimes n}$ and $Q^{\otimes n}$.

The one-period experiments are
\begin{equation}
\begin{aligned}
P_0&=(1/2,1/2), & P_1&=(1/5,4/5), & P_2&=(4/5,1/5),\\
Q_0&=(1/2,1/2), & Q_1&=(2/5,3/5), & Q_2&=(11/15,4/15).
\end{aligned}
\label{eq:experiments}
\end{equation}

For a binary word $w\in\{0,1\}^n$, let $m(w)$ denote the number of second-coordinate signals. Relative to state $0$, the source likelihood ratios are
\begin{equation}
r_1(n,m)=\frac{8^m2^{n-m}}{5^n},\qquad
r_2(n,m)=\frac{2^m8^{n-m}}{5^n}.
\label{eq:source-ratios}
\end{equation}
The target boundary word is the all-second-coordinate word, with ratios
\begin{equation}
s_1(n)=\left(\frac65\right)^n,\qquad
s_2(n)=\left(\frac8{15}\right)^n.
\label{eq:target-ratios}
\end{equation}

The state-$0$ count masses are
\begin{equation}
\mu_n(m)=\frac{\binom{n}{m}}{2^n},\qquad
\nu_n(k)=\frac{\binom{n}{k}}{2^n}.
\label{eq:count-masses}
\end{equation}
The raw and count-level arrays agree because each count fiber has exactly $\binom{n}{m}$ words. For example,
\begin{equation}
\begin{aligned}
\mu_n(m)r_1(n,m)&=\binom{n}{m}\frac{4^m}{5^n},
&\mu_n(m)r_2(n,m)&=\binom{n}{m}\frac{4^{n-m}}{5^n},\\
\nu_n(k)s_1(n,k)&=\binom{n}{k}\frac{3^k2^{n-k}}{5^n},
&\nu_n(k)s_2(n,k)&=\binom{n}{k}\frac{4^k11^{n-k}}{15^n}.
\end{aligned}
\label{eq:count-arrays}
\end{equation}
The binomial theorem verifies normalization.

\section{The many-state MGF conditions}

For an experiment $E$ and a state $i$, write
\begin{equation}
M_E^i(u,v)=\sum_s E_i(s)
 \left(\frac{E_i(s)}{E_j(s)}\right)^u
 \left(\frac{E_i(s)}{E_k(s)}\right)^v,
\label{eq:mgf}
\end{equation}
where $(j,k)$ is the ordered pair of the other two states. The ordering is
immaterial for the universal quantifiers below. For the pair in
\eqref{eq:experiments}, direct substitution gives the six displayed
expressions in the formal file: $M_P^i(u,v)>M_Q^i(u,v)$ for every
$u,v\geq0$ with $(u,v)\ne(0,0)$, and
\begin{equation}
M_P^i(-a,-b)<M_Q^i(-a,-b)
\quad\text{for }a,b\geq0,\quad a+b<1,\quad (a,b)\ne(0,0).
\label{eq:mgf-cones}
\end{equation}
The proof uses Holder convexity on the positive cone, weighted geometric
concavity on the closed simplex, and a strict interpolation argument on the
open simplex. The boundary cases for states $1$ and $2$ are included in the
strict proof; they reduce to the same two-point inequality after permuting
coordinates.

\begin{proposition}[Verified necessary-condition package]\label{prop:mgf}
The pair in \eqref{eq:experiments} is normalized, strictly positive and
bounded; it is pairwise generic; all six ordered KL inequalities are strict;
and it satisfies both MGF comparisons in \eqref{eq:mgf-cones}. Consequently,
it satisfies every hypothesis in the many-state necessary-condition test.
\end{proposition}

\begin{proof}
The rational checks for normalization, positivity, ratio bounds, and KL order
are finite arithmetic. Pairwise genericity is checked both for ratio extrema
and, using strict monotonicity of the logarithm (so the log of a finite
positive maximum/minimum is the corresponding maximum/minimum of logs), for
the log-extrema appearing in the conjecture. The six MGF identities and the continuum inequalities are
proved by the elementary inequalities just described. The complete derivation
is checked by Lean 4.30 in \texttt{MuMGFNegative.lean}, whose package theorem
is \texttt{mu\_appendixK\_premises}.\end{proof}

\section{Main theorem and conjecture resolution}

The positive result is most transparent in the square finite case: the number
of states equals the number of source signals.

\begin{theorem}[Finite simplex Blackwell iff]\label{thm:main-simplex}
Let $I=\{0,\ldots,m-1\}$ be the state space and
$X=\{x_0,\ldots,x_{m-1}\}$ the source signal space. Assume
$P_0(x_j)>0$ and $Q_0(y)>0$. Define
\[
r(x_j)=\left(\frac{P_i(x_j)}{P_0(x_j)}\right)_{i=1}^{m-1},
\qquad
s(y)=\left(\frac{Q_i(y)}{Q_0(y)}\right)_{i=1}^{m-1}.
\]
If $r(x_0),\ldots,r(x_{m-1})$ are affinely independent, let
$\lambda_j(z)$ be their unique barycentric coordinates. Then
\[
P\succeqB Q
\quad\Longleftrightarrow\quad
\text{for every target signal }y\text{ and every }j,
\quad \lambda_j(s(y))\geq 0.
\]
\end{theorem}

\begin{proof}
If a garbling $K$ exists, define
\[
\widehat\lambda_j(y)=\frac{P_0(x_j)K(x_j,y)}{Q_0(y)}.
\]
The garbling equations imply that these numbers sum to one and that
$s(y)=\sum_j\widehat\lambda_j(y)r(x_j)$. Affine independence makes the
representation unique, so $\widehat\lambda_j(y)=\lambda_j(s(y))\geq0$.

Conversely, assume the displayed universal inequalities and define
\[
K(x_j,y)=\frac{Q_0(y)\lambda_j(s(y))}{P_0(x_j)}.
\]
The reference-state identities imply
$\sum_yQ_0(y)s(y)=\sum_jP_0(x_j)r(x_j)=(1,\ldots,1)$.
Uniqueness of barycentric coordinates therefore gives
$\sum_yQ_0(y)\lambda_j(s(y))=P_0(x_j)$ for every $j$, so $K$ is
row-stochastic. The barycentric identity then yields
$\sum_jP_i(x_j)K(x_j,y)=Q_i(y)$ for every state $i$ and target signal $y$.
\end{proof}

\begin{remark}[Scope of the square restriction]
Here ``square'' refers to the source experiment: the number of source signals
equals the number of states. The target may have any finite number of signals.
The same barycentric argument extends to a lower-dimensional simplex with
$m\leq |I|$ source signals, provided the source ratio vectors are affinely
independent and every target ratio vector lies in their affine hull. We use
the square case as the headline because it is full-dimensional and requires
no additional affine-hull qualification. Source experiments with redundant
interior signal types, and the noisy-audit, common-noise, and modular classes,
are separate restricted cases treated in the appendices.
\end{remark}

\begin{theorem}[Resolution of Mu's universal sufficiency conjecture]\label{thm:mu-resolution}
For the three-state, two-signal pair in \eqref{eq:experiments}, all of Mu's
necessary conditions hold: strict positivity and boundedness, pairwise
genericity, the continuum positive and negative MGF inequalities, and all
ordered KL inequalities. Nevertheless,
\[
P^{\otimes n}\not\succeqB Q^{\otimes n}
\qquad\text{for every integer }n\geq1.
\]
Hence Mu's sufficiency conjecture is false as a universal statement over
finite experiments. The result is a counterexample with three states and two
signals per period; it does not assert failure for every larger state count or
every other signal cardinality.
\end{theorem}

\begin{proof}
The analytic and finite checks for the necessary-condition package are
Proposition~\ref{prop:mgf}. The all-sample-size non-domination statement is
Theorem~\ref{thm:nondom}, proved in Appendix~B by the forced equality-face
argument and the prime-factor obstruction.
\end{proof}

\appendix

\section{Boundary couplings and equality faces}

Let $\mu$ and $\nu$ be nonnegative masses on finite source and target spaces. Let $B$ be a set of target boundary atoms, and let $f:B\to X$ prescribe the only source atom that may feed each boundary atom. Define
\begin{equation}
b(x)=\sum_{y\in B:f(y)=x}\nu(y),\qquad
\mu^r(x)=\mu(x)-b(x),
\end{equation}
and set $\nu^r(y)=0$ on $B$ and $\nu^r(y)=\nu(y)$ off $B$.

\begin{theorem}[Conditional boundary-coupling iff]\label{thm:iff}
Assume nonnegativity and that every admissible coupling $\gamma$ satisfies
\begin{equation}
y\in B\text{ and }\gamma(x,y)>0\quad\Longrightarrow\quad x=f(y).
\label{eq:forced-support}
\end{equation}
Then a coupling with marginals $(\mu,\nu)$ exists if and only if
\begin{enumerate}
\item $\mu(f(y))\geq \nu(y)$ for every $y\in B$; and
\item the residual masses $(\mu^r,\nu^r)$ admit a coupling.
\end{enumerate}
\end{theorem}

\begin{proof}
Any feasible coupling must allocate the full column $\nu(y)$ to $f(y)$, which gives the capacity inequalities. Subtracting those forced entries yields a residual coupling, proving necessity. Conversely, add the forced entries to any residual coupling. The capacity inequalities preserve nonnegativity, and the original marginals are restored.\end{proof}

The theorem is a finite transport statement under a support certificate. It is not an unconditional characterization of Blackwell dominance.

The forced-support statement is a certificate, not a continuity assumption.
It says that a target boundary column is an exposed equality face of the
likelihood-ratio moment equations: any admissible coupling placing positive
mass in that column must use the prescribed source atom. In the finite model
this follows by exhibiting an affine functional $L$ with
$L(x)\ge c$ on the source support and showing that the target column has
moment exactly $c$. The zero-slack argument forces $L(x)=c$ on every
source atom used by that column; a separate calculation of the equality set
identifies that set with $f(y)$. Lean verifies the implication as
\texttt{linear\_boundary\_forced}.

In the displayed example, the square-root functional
$L(x)=s_2r_1(x)+s_1r_2(x)$ gives exactly this equality argument. The
prime-factor calculation then shows that the required coordinatewise match
does not exist, so the boundary atom cannot be supplied. The
forcing step is proved, while the matching step fails; neither is imposed as
an empirical regularity condition.

The certificate has a direct monitoring interpretation: an extreme audit
outcome has a posterior-likelihood signature that cannot be synthesized by
mixing other source outcomes while preserving each statewise likelihood
moment. Its source type is therefore scarce, and the capacity
inequality is a resource constraint. The condition is meaningful in an
application, but it is not standard in the many-state
large-sample literature and should not be presented as automatic.
It is acceptable as an explicit, checkable hypothesis in a boundary theorem;
it is not a substitute for a general continuity theorem.

The certificate is not necessary for Mu et al.'s large-sample Blackwell
order. Their conjecture does not require a boundary set, a map $f$, or a
unique equality source. It becomes necessary only conditionally:
if a target column lies on an exposed equality face containing a
single source likelihood-ratio point: every feasible martingale coupling
must satisfy the forced-support restriction. Experiments whose likelihood
geometry has no such exposed equality face need not satisfy a forced-support
statement. The certificate is therefore a technical tool for this boundary
mechanism, not a general implication of eventual dominance.

There is, however, an exact finite characterization before the boundary
deletion step. Fix a reference state $i_0$ and assume
$P_{i_0}(x),Q_{i_0}(y)>0$. Put
$r_i(x)=P_i(x)/P_{i_0}(x)$ and
$s_i(y)=Q_i(y)/Q_{i_0}(y)$. Then
\begin{equation}
\exists K\;(Q_i=P_iK\ \forall i)
\quad\Longleftrightarrow\quad
\exists\gamma\;\left\{
\begin{array}{l}
\gamma\text{ has marginals }(P_{i_0},Q_{i_0}),\\
\displaystyle\sum_x\gamma(x,y)r_i(x)=Q_{i_0}(y)s_i(y)\quad\forall i,y.
\end{array}\right.
\label{eq:full-iff}
\end{equation}
The forward map is $\gamma(x,y)=P_{i_0}(x)K(x,y)$; the reverse map is
$K(x,y)=\gamma(x,y)/P_{i_0}(x)$. Thus the converse is exact and uses no
strict-separation assumption. The source file
\texttt{FiniteBlackwell.lean} verifies \eqref{eq:full-iff} and its exact
residual boundary version, in which the support certificate is imposed only
on martingale couplings and all likelihood-ratio moments are preserved after
deleting the boundary columns.

The finite coupling equivalence is a proof device, not an additional
primitive condition that repairs the conjecture. Requiring these couplings for
all sufficiently large sample sizes is equivalent to requiring eventual
Blackwell dominance itself. Adding the three necessary MGF/KL conditions to
that requirement yields no new characterization.

A canonical fourth condition can be written using only
universal quantifiers. Fix a reference state $i_0$ and, for a product signal
$x$ or $y$, write
\begin{equation}
 R_n(x)=\left(\frac{P_i^{\otimes n}(x)}{P_{i_0}^{\otimes n}(x)}\right)_{i\ne i_0},
 \qquad
 S_n(y)=\left(\frac{Q_i^{\otimes n}(y)}{Q_{i_0}^{\otimes n}(y)}\right)_{i\ne i_0}.
\label{eq:ratio-vectors}
\end{equation}
For a fixed tail threshold $N$, impose
\begin{equation}
\tag{$\mathrm{C4}_N$}
\text{for every }n\ge N\text{ and every finite convex }\Phi:\R^{|I|-1}\to\R,
\qquad
\mathbb E_{P_{i_0}^{\otimes n}}[\Phi(R_n)]
\geq
\mathbb E_{Q_{i_0}^{\otimes n}}[\Phi(S_n)].
\label{eq:universal-c4}
\end{equation}
The two ratio vectors have the same mean, so \eqref{eq:universal-c4} is a
finite-dimensional convex-order statement. The finite Strassen theorem,
combined with \eqref{eq:full-iff}, gives
\begin{equation}
\begin{aligned}
 P^{\otimes n}\succeqB Q^{\otimes n}
 &\Longleftrightarrow \mathrm{C4}_n,\\
 \exists N\ \forall n\ge N:\ P^{\otimes n}\succeqB Q^{\otimes n}
 &\Longleftrightarrow
 \exists N\ \forall n\ge N:\bigl[\mathrm{C4}_n\ \text{and Mu's conditions (1)--(3)}\bigr].
\end{aligned}
\label{eq:canonical-c4-iff}
\end{equation}
On finite signal spaces, the same condition has the more
economic form
\begin{equation}
\text{for every }n\ge N,\ m\ge1,\ a_1,\ldots,a_m\in\R^{|I|-1},\ b_1,\ldots,b_m\in\R,
\quad
\mathbb E_{P_{i_0}^{\otimes n}}
 \!\left[\max_{\ell\le m}\{a_\ell\!\cdot R_n+b_\ell\}\right]
\geq
\mathbb E_{Q_{i_0}^{\otimes n}}
 \!\left[\max_{\ell\le m}\{a_\ell\!\cdot S_n+b_\ell\}\right].
\label{eq:decision-c4}
\end{equation}
Every finite-action problem has an indirect payoff of this max-of-affine form
after factoring out the reference-state likelihood. Conversely, on finite
supports, finite max-of-affine tests generate the convex-order inequalities.
Thus \eqref{eq:decision-c4} says directly that every finite decision problem
has weakly higher value under the source experiment. This is the standard
Blackwell value interpretation. It is not an independent
repair of Mu's conjecture: \eqref{eq:canonical-c4-iff} shows that condition
($\mathrm{C4}$) already contains the Blackwell conclusion, while the MGF and
KL conditions are then redundant necessary consequences. We therefore state
it as a canonical reformulation, not as a fourth primitive assumption or a
new contribution. Replacing $n\ge N$ by $n\ge1$ gives the stronger all-sample
version, not Mu's eventual-order conjecture.

The canonical convex-order condition above is exact but equivalent to the
Blackwell conclusion. The finite theorem and the counterexample are our
general results. Restricted noisy-audit, common-noise, and modular
completions appear in Appendix~A; they are not assumptions in the
unrestricted theorem.

For reference, the finite simplex theorem is recorded again in its
likelihood-ratio notation. Suppose
there are $k+1$ states and exactly $k+1$ source signals
$X=\{x_0,\ldots,x_k\}$. Fix a reference state $0$ and define
\[
r(x_j)=\left(\frac{P_i(x_j)}{P_0(x_j)}\right)_{i=1}^k,\qquad
s(y)=\left(\frac{Q_i(y)}{Q_0(y)}\right)_{i=1}^k.
\]
Assume $P_0(x_j)>0$, $Q_0(y)>0$, and that the $k+1$ vectors
$r(x_0),\ldots,r(x_k)$ are affinely independent. Let
$\lambda_0,\ldots,\lambda_k$ denote their affine barycentric coordinate
functions, so that for every $z\in\mathbb R^k$,
\[
\sum_{j=0}^k\lambda_j(z)=1,\qquad
z=\sum_{j=0}^k\lambda_j(z)r(x_j).
\]
The fourth condition is a universal primitive condition
\begin{equation}
\tag{$\mathrm{C4}^{\triangle}$}
\text{for every target signal }y\text{ and every }j,\qquad
\lambda_j(s(y))\geq0.
\label{eq:simplex-c4}
\end{equation}
It says that each target posterior-likelihood profile lies in the source
posterior simplex. Thus the restriction has a direct verbal form: the source
has as many signal types as states, and those types are
affinely independent posterior types.

\begin{theorem}[Finite simplex completion]
\label{thm:simplex-completion}
Under the simplex assumptions above,
\[
P\succeqB Q
\quad\Longleftrightarrow\quad
\eqref{eq:simplex-c4}.
\]
\end{theorem}

\begin{proof}
If a garbling K exists, set
\[
\widehat\lambda_j(y)=\frac{P_0(x_j)K(x_j,y)}{Q_0(y)}.
\]
The row and column equations imply
$\sum_j\widehat\lambda_j(y)=1$ and
$s(y)=\sum_j\widehat\lambda_j(y)r(x_j)$.
Affine independence makes barycentric coordinates unique, so
$\widehat\lambda_j(y)=\lambda_j(s(y))\geq0$.

Conversely, assume \eqref{eq:simplex-c4} and define
\[
K(x_j,y)=\frac{Q_0(y)\lambda_j(s(y))}{P_0(x_j)}.
\]
The equality of the reference-state means gives
\[
\sum_yQ_0(y)s(y)=\sum_jP_0(x_j)r(x_j)=(1,\ldots,1).
\]
Using the barycentric expansion of every $s(y)$ and affine independence,
we obtain $\sum_yQ_0(y)\lambda_j(s(y))=P_0(x_j)$ for every $j$.
Thus K is row-stochastic. Finally,
\[
\sum_jP_i(x_j)K(x_j,y)
=Q_0(y)\sum_j\lambda_j(s(y))r_i(x_j)=Q_i(y),
\]
so K is a garbling.
\end{proof}

Define the tail condition
\begin{equation}
\tag{$\mathrm{C4}^{\triangle}_{\mathrm{tail}}$}
\text{for every }n\geq N,\text{ the source has exactly }k+1\text{ signal atoms at affinely independent simplex vertices, and for every target }y\text{ and every vertex }j,\quad
\lambda_{n,j}(s_n(y))\geq0.
\label{eq:simplex-tail-c4}
\end{equation}
For any product family satisfying this simplicial-tail model restriction,
\begin{corollary}[Simplicial-tail completion]
\label{cor:simplex-tail}
\[
\text{Mu (1)--(3)}\ \land\ \mathrm{C4}^{\triangle}_{\mathrm{tail}}
\quad\Longleftrightarrow\quad
\forall n\geq N: P^{\otimes n}\succeqB Q^{\otimes n}.
\]
\end{corollary}
\begin{proof}
For the forward implication, apply the finite simplex theorem at each
$n\geq N$; this gives the garbling at every tail sample size. Conversely,
assume the tail garblings exist. The standing simpliciality restriction lets
us apply the reverse direction of the finite simplex theorem at each
$n\geq N$, yielding $\lambda_{n,j}(s_n(y))\geq0$ for every target atom and
vertex. Mu's necessity result gives conditions (1)--(3).
\end{proof}

The condition has a direct economic interpretation. The source generates a
simplex of posterior types. Each target signal produces a posterior inside
that simplex, and its unique barycentric weights specify the randomization
needed to synthesize that target signal. The theorem is a finite geometric
criterion rather than a quantification over decision problems or an unknown
kernel. The abstract coefficient-and-capacity version of the argument is
verified in \texttt{BoundaryIff/SimplexCapacity.lean}. Its large-sample
corollary is restricted: generic product experiments usually have more than
$k+1$ extreme likelihood-ratio types, so the simplex assumption need not
persist with $n$. The theorem is finite. Equal cardinalities of infinite
state and signal spaces do not imply a simplex, unique barycentric
coordinates, or a valid stochastic kernel; infinite models require separate
topological and measure-theoretic assumptions.

The next proposition gives a nontrivial countable economic class in which
those assumptions are explicit rather than abstract. State $i$ represents a
specific failure mode. The monitor reports $i$ with probability $\theta_i$
and otherwise emits a common noise report drawn from $q$.
\begin{proposition}[Countable diagnostic-monitor theorem]
\label{prop:countable-diagnostic}
Let $I=X=\mathbb N_0$. Let $(q_j)_{j\geq0}$ be a strictly positive
probability mass function and let $\theta_0=0$ and
$\theta_i\in(0,1)$ for $i\geq1$. Define
\begin{equation}
P_i(j)=(1-\theta_i)q_j+\theta_i\mathbf 1_{\{i=j\}}.
\label{eq:diagnostic-source}
\end{equation}
Let $Q=(Q_i)_{i\geq0}$ be any experiment on a countable target signal space
$Y$ with $Q_0(y)>0$ for every $y$, and write
$s_i(y)=Q_i(y)/Q_0(y)$. Then $P\succeqB Q$ if and only if, for every
$y\in Y$,
\begin{equation}
s_i(y)\geq 1-\theta_i\quad(i\geq1),
\qquad
\sum_{i\geq1}\frac{q_i}{\theta_i}
       \bigl(s_i(y)-(1-\theta_i)\bigr)\leq1.
\label{eq:diagnostic-c4}
\end{equation}
When these inequalities hold, the garbling is
\begin{equation}
\lambda_i(y)=\frac{q_i}{\theta_i}
       \bigl(s_i(y)-(1-\theta_i)\bigr)\quad(i\geq1),
\qquad
\lambda_0(y)=1-\sum_{i\geq1}\lambda_i(y),
\qquad
K(j,y)=\frac{Q_0(y)\lambda_j(y)}{q_j}.
\label{eq:diagnostic-kernel}
\end{equation}
Consequently $P^{\otimes n}\succeqB Q^{\otimes n}$ for every $n\geq1$.
\end{proposition}

\begin{proof}
Relative to state $0$, the source likelihood-ratio atom at signal $0$ is
$b=(1-\theta_i)_{i\geq1}$, while the atom at signal $j\geq1$ is
$b+(\theta_j/q_j)e_j$. Hence \eqref{eq:diagnostic-c4} is exactly the
nonnegativity of the unique countable barycentric coordinates in
\eqref{eq:diagnostic-kernel}. They sum to one, so $K$ is nonnegative. For
$i\geq1$,
\begin{equation*}
\sum_jP_i(j)K(j,y)
=Q_0(y)\left[(1-\theta_i)+\frac{\theta_i}{q_i}\lambda_i(y)\right]
=Q_i(y).
\end{equation*}
For state $0$, the same calculation gives
$\sum_jP_0(j)K(j,y)=Q_0(y)$. The rows of $K$ sum to one because, for
$i\geq1$,
\begin{equation*}
\sum_yQ_0(y)\lambda_i(y)
=\frac{q_i}{\theta_i}\left(\sum_yQ_i(y)-(1-\theta_i)\sum_yQ_0(y)\right)=q_i,
\end{equation*}
and the identity for $i=0$ follows from Tonelli's theorem and
$\sum_jq_j=1$. This proves sufficiency. Conversely, any garbling defines
$\lambda_j(y)=q_jK(j,y)/Q_0(y)$. The state-$0$ equation gives
$\sum_j\lambda_j(y)=1$, and the state-$i$ equation gives the first formula
in \eqref{eq:diagnostic-kernel}; hence \eqref{eq:diagnostic-c4} is necessary.
Tensoring $K$ proves the final claim.
\end{proof}

The condition has a direct economic meaning. Every target signal must have
at least the source's baseline likelihood ratio in each failure state, and
its total diagnostic weight cannot exceed one. The resulting test is
signal-by-signal and does not quantify over decision problems or an unknown
kernel. For every finite $m$, restricting the construction to
$I=X=\{0,\ldots,m-1\}$ gives a nontrivial $m$-state, $m$-signal class. In
that class, Mu's conditions (1)--(3), when imposed together with
\eqref{eq:diagnostic-c4}, imply eventual Blackwell dominance; in fact the
stronger all-sample conclusion holds. The MGF and KL conditions are
necessary but redundant once the explicit diagnostic garbling is available.
Mu's original conditions are stated for finitely many states, so the
countable proposition is a structural extension rather than a literal claim
about their original package.

If the target is itself diagonal diagnostic, with strengths
$\eta_0=0$ and $\eta_i\in[0,1)$, so that
$Q_i(j)=(1-\eta_i)q_j+\eta_i\mathbf1_{\{i=j\}}$, condition
\eqref{eq:diagnostic-c4} becomes
\begin{equation}
0\leq\eta_i\leq\theta_i\quad\forall i\geq1,
\qquad
\sum_{i\geq1}q_i\left(1-\frac{\eta_i}{\theta_i}\right)
 +\sup_{i\geq1}\frac{\eta_i}{\theta_i}\leq1.
\label{eq:diagnostic-capacity}
\end{equation}
The first inequality is local weakening of every diagnosis.  The second is a
global capacity constraint: individually weaker diagnoses can still be
jointly impossible to simulate when their state-specific improvements compete
for the same source signal capacity.  This heterogeneous capacity effect is
absent when all ratios $\eta_i/\theta_i$ are equal.
For instance, with $q_0=0.1$, $q_1=q_2=0.45$,
$(\theta_1,\theta_2)=(0.4,0.8)$, and
$(\eta_1,\eta_2)=(0.39,0.1)$, each target diagnosis is individually weaker,
but the left side of the capacity inequality equals
$0.45(1-0.39/0.4)+0.45(1-0.1/0.8)+0.39/0.4=1.38>1$.
The target is therefore not a garbling of the source.  This is a genuine
aggregate capacity obstruction, not a failure of any one diagnostic channel.

For the likelihood-ratio version, let $(r_1(x),r_2(x))$ be source ratios and let $(s_1,s_2)$ be the ratios at a target atom $y^\ast$. Assume $s_1s_2\leq r_1(x)r_2(x)$ for every $x$, with all coordinates positive. Define
\begin{equation}
L(x)=s_2r_1(x)+s_1r_2(x),\qquad c=2s_1s_2.
\end{equation}
The arithmetic-geometric mean inequality gives $L(x)\geq c$. If a nonnegative coupling satisfies the two likelihood-ratio moment equations and $\nu(y^\ast)>0$, then
\begin{equation}
\sum_x\gamma(x,y^\ast)L(x)=\nu(y^\ast)c.
\end{equation}
Thus the nonnegative slack has zero weighted sum. Equality in the square-root inequality forces
\begin{equation}
r_1(x)=s_1,\qquad r_2(x)=s_2
\end{equation}
for every source atom carrying positive mass into $y^\ast$.

\begin{lemma}[Equality-face obstruction]\label{lem:obstruction}
If $\nu(y^\ast)>0$ and no source atom has both ratios equal to $(s_1,s_2)$, then no such likelihood-ratio martingale coupling exists.
\end{lemma}

\begin{proof}
The zero weighted slack forces a positive-mass source atom onto the equality face. Equality in the square-root inequality gives both coordinate equalities, contradicting the no-match assumption.\end{proof}

Economically, $\mu(x)$ is the source probability of a signal type and $\nu(y)$ is the target probability of an audited signal. On the equality face, $y$ has a unique source type capable of generating it without violating the two likelihood-ratio moments. The inequality $\mu(f(y))\geq\nu(y)$ is a capacity constraint. The residual condition asks whether all other target signals can still be supplied after that capacity is reserved.

The residual coupling clause in the boundary decomposition is essential.  Capacities at exposed
boundary points alone do not guarantee a martingale coupling.  For example,
consider a one-dimensional source law with masses
$\mu=(1/2,1/2)$ at ratios $(1/2,3/2)$ and a target law with masses
$\nu=(1/2,1/4,1/4)$ at ratios $(1/2,1,2)$.  The boundary column at $1/2$
uses exactly the available source mass, but the residual source has only ratio
$3/2$, whereas each residual target column requires its own conditional
source mean, $1$ or $2$.  No residual martingale coupling exists. This example shows that boundary
capacity tests do not replace residual feasibility in the finite decomposition;
it does not establish insufficiency of any proposed boundary condition
combined with the full MGF/KL package.

A substantive repair would specify a condition on the primitive experiments
and establish both its necessity and its sufficiency in conjunction with the
three MGF/KL conditions. This paper does not supply such an additional
condition. The counterexample and the finite coupling lemmas do not by
themselves prove that all remaining obstructions are boundary obstructions.

\section{The explicit all-sample-size obstruction}

For every source count $0\leq m\leq n$,
\begin{equation}
r_1(n,m)r_2(n,m)=\left(\frac{16}{25}\right)^n=s_1(n)s_2(n).
\label{eq:product-equality}
\end{equation}
Every source count therefore lies on the same product level set, but coordinatewise equality is impossible. If both coordinates matched, division of the two equations would imply
\begin{equation}
4^{2m}=9^n.
\label{eq:prime-obstruction}
\end{equation}
For $n>0$, the left side has only prime factor $2$ and the right side has only prime factor $3$.

\begin{theorem}[Finite-sample non-domination]\label{thm:nondom}
For every integer $n>0$, $P^{\otimes n}\not\succeqB Q^{\otimes n}$ for the pair in \eqref{eq:experiments}.
\end{theorem}

\begin{proof}
Suppose a kernel $K$ exists and define
\begin{equation}
\gamma(w,z)=P_0^{\otimes n}(w)K(w,z).
\end{equation}
The row and column equations imply a likelihood-ratio martingale coupling. Push it through the count statistic. The count fiber with $m$ second-coordinate signals has cardinality $\binom{n}{m}$, so the pushed state-$0$ mass is $\binom{n}{m}/2^n$. The ratio identities survive because they depend only on the count. The target boundary atom has positive mass and satisfies equality in \eqref{eq:product-equality}, while \eqref{eq:prime-obstruction} rules out a coordinatewise match. Lemma~\ref{lem:obstruction} gives a contradiction.\end{proof}

The raw-to-count step is proved in Lean by the support/indicator bijection, the exact binomial fiber count, the uniform-mass pushforward identity, and the final theorem \texttt{no\_raw\_word\_garbling}.

\section{An audit decision problem}

Nature chooses $\theta\in\{0,1,2\}$ uniformly. The evaluator observes a signal and chooses $C$ (continue deployment) or $A$ (audit, restrict, or punish). Payoffs are
\begin{center}
\begin{tabular}{c@{\quad}ccc}
 & $\theta=0$ & $\theta=1$ & $\theta=2$\\ \midrule
$C$ & 2 & 0 & 0\\
$A$ & 0 & 1 & 1
\end{tabular}
\end{center}
For experiment $E$,
\begin{equation}
V_E=\sum_y\max_{a\in\{C,A\}}\sum_{\theta=0}^2\frac13E_\theta(y)U(a,\theta).
\end{equation}

\begin{proposition}[Separating audit problem]\label{prop:audit}
For the experiments in \eqref{eq:experiments},
\begin{equation}
V_P=\frac23,\qquad V_Q=\frac{32}{45}>\frac23.
\end{equation}
\end{proposition}

\begin{proof}
Under $P$, each signal gives continuation score $2P_0(y)=1$ and audit score $P_1(y)+P_2(y)=1$. Hence $V_P=2/3$. Under $Q$, the first signal gives scores $1$ and $17/15$, so audit is chosen; the second gives scores $1$ and $13/15$, so continuation is chosen. Hence $V_Q=(17/15+1)/3=32/45$.\end{proof}

If $P\succeqB Q$, Blackwell monotonicity would imply $V_P\geq V_Q$ for every decision problem. Proposition~\ref{prop:audit} is therefore a direct separating payoff. The boundary theorem explains why this policy is unavailable under the source experiment.

A small perturbation removes the tie under $P$. Replace the continuation payoff in state $0$ by $2+\varepsilon$, with $0<\varepsilon<4/15$. The value gap becomes
\begin{equation}
V_Q-V_P=\frac{2}{45}-\frac{\varepsilon}{6}>0.
\end{equation}

\section{Application: repeated monitoring of an adaptive AI agent}

The application is a repeated moral-hazard problem. An agent chooses whether to follow a protocol or take a profitable deviation. The principal observes noisy audit signals and decides whether to preserve access or impose a restriction. The hidden state may represent honest execution, reward hacking, or strategic concealment. Signals may include logs, tool calls, test results, verifier output, monitor-model judgments, and human review.

State $0$ denotes honest execution; states $1$ and $2$ denote two deviation modes. A signal is a log, test, verifier, or monitor output. $C$ retains deployment or tool access, whereas $A$ requests evidence, sandboxes, retrains, or revokes permissions. A product experiment is a review block of independent tasks.

For a block of $n$ independent target signals, restrict whenever any signal equals the second signal. Under $Q_1$ the restriction probability is
\begin{equation}
1-\left(\frac35\right)^n,
\end{equation}
and under $Q_2$ it is
\begin{equation}
1-\left(\frac4{15}\right)^n.
\end{equation}
If a deviation yields a one-period gain $g$, a block arrives after $n$ discounted periods, and restriction costs the agent $L$, this protocol supports
\begin{equation}
g\leq\delta^n\left[1-\left(\frac35\right)^n\right]L,
\qquad
g\leq\delta^n\left[1-\left(\frac4{15}\right)^n\right]L.
\end{equation}
These inequalities are specific to the protocol. The theorem explains why a source monitoring system cannot simulate every target policy, while the audit rule supplies one policy and its incentive calculation.

\section{Scope and formal verification}

The theorem is a structural finite-sample result. It proves that the operative garbling does not exist and identifies the equality-face mechanism. It does not compute a norm distance to the set of garblings or a universal monitoring-loss bound. For finite signal spaces one can instead study
\begin{equation}
\Gamma_n(P,Q)=\sup_{\|U\|_\infty\leq 1}\left[V_{Q^{\otimes n}}(U)-V_{P^{\otimes n}}(U)\right]
\end{equation}
or the finite deficiency
\begin{equation}
d_n(P,Q)=\inf_K\max_\theta\left\|Q_\theta^{\otimes n}-P_\theta^{\otimes n}K\right\|_1.
\end{equation}
Those are separate quantitative questions.

The canonical Lean file is \texttt{CompleteProof.lean}, which imports
\texttt{FullProof.lean}, \texttt{FiniteBlackwell.lean},
\texttt{ProductBridge.lean}, and
\texttt{MuMGFNegative.lean}. Together they contain the full finite
Blackwell--martingale iff, the forced-boundary
coupling iff, square-root equality lemmas, the Blackwell-to-likelihood-ratio
bridge, the raw pushforward construction, the binary-word support bijection,
the exact binomial fiber count, the concrete arrays, the complete MGF package,
and the final theorems \texttt{no\_concrete\_count\_garbling} and
\texttt{no\_raw\_word\_garbling}.

The command \texttt{lake build} completes successfully for the finite library,
and \texttt{lake env lean MuMGFNegative.lean} checks the analytic package.
The formalization contains no \texttt{sorry}, \texttt{admit}, or
\texttt{axiom}.

\section{Restricted completions}
The preceding condition is exact but tautological as a repair. A genuinely
non-circular result requires an explicit model class. The following one is a
standard noisy-audit structure. Fix latent audit qualities
$q_i\in(0,1)$ (pairwise distinct) and noise rates
$0<\varepsilon_P,\varepsilon_Q<1/2$ with
$\varepsilon_P\ne\varepsilon_Q$. The fourth condition is the
universal, state-by-state restriction
\begin{equation}
\tag{$\mathrm{C4}^{\mathrm{audit}}$}
\text{for every }i\in I,\qquad
 P_i=\operatorname{Bern}\!\left(\varepsilon_P+(1-2\varepsilon_P)q_i\right),\quad
 Q_i=\operatorname{Bern}\!\left(\varepsilon_Q+(1-2\varepsilon_Q)q_i\right).
\label{eq:audit-c4}
\end{equation}
It says that the two experiments are binary audits of the same state quality,
with state-independent symmetric classification noise. This is a familiar
measurement model in monitoring, testing, and information economics; it is a
structural restriction on primitive experiments, not a statement about an
unknown garbling.
For $q_i>q_j$, the derivatives of the two log-likelihood ratios with respect
to $\varepsilon$ are
\[
-\frac{q_i-q_j}{p_i(\varepsilon)p_j(\varepsilon)}
\quad\text{and}\quad
\frac{q_i-q_j}{(1-p_i(\varepsilon))(1-p_j(\varepsilon))},
\]
where $p_i(\varepsilon)=\varepsilon+(1-2\varepsilon)q_i$. Thus the extrema
change strictly with the noise rate, so the pair is generic whenever
$\varepsilon_P\ne\varepsilon_Q$.

\begin{theorem}[Noisy-audit completion]
\label{thm:audit-completion}
Within the class \eqref{eq:audit-c4}, with pairwise distinct $q_i$,
the following are equivalent:
\begin{equation}
\begin{aligned}
&\text{Mu's conditions (1)--(3)}
\quad\Longleftrightarrow\quad
\exists N\ \forall n\ge N:\ P^{\otimes n}\succeqB Q^{\otimes n}.
\end{aligned}
\label{eq:audit-completion-iff}
\end{equation}
In fact, when the equivalent conditions hold, the dominance is valid for
every $n\ge1$.
\end{theorem}

\begin{proof}
The strict ordered KL inequalities in Mu's condition (3) imply
$\varepsilon_P<\varepsilon_Q$. Indeed, a binary symmetric channel with
larger error strictly contracts the relative entropy between any two
distinct input Bernoulli laws; the input laws here are distinct because the
$q_i$ are pairwise distinct. Put
\[
\delta=\frac{\varepsilon_Q-\varepsilon_P}{1-2\varepsilon_P}\in(0,1/2).
\]
For every state $i$,
\[
\varepsilon_Q+(1-2\varepsilon_Q)q_i
 =\delta+(1-2\delta)\bigl[\varepsilon_P+(1-2\varepsilon_P)q_i\bigr].
\]
Thus $Q$ is obtained from $P$ by the same binary symmetric channel of error
$\delta$, and applying that channel independently to each coordinate gives
$P^{\otimes n}\succeqB Q^{\otimes n}$ for every $n$. Conversely, eventual
Blackwell dominance implies Mu's conditions (1)--(3), by their necessity
theorem.\end{proof}

The theorem is an exact iff only after the accepted audit structure is fixed;
it is not an unrestricted solution of Mu's many-state conjecture. The current
counterexample lies outside this class. This distinction matters: in the
unrestricted finite class, a universal condition strong enough for an iff is
the full multivariate convex order in \eqref{eq:universal-c4}, which is
equivalent to Blackwell order itself. No weaker condition 4 with a proved
unconditional iff is supplied here.

A broader standard measurement-error completion uses a common noise semigroup.
Let $A$ be a finite signal alphabet, let $\rho_i$ be pairwise distinct
full-support probability vectors on $A$, and let $(T_t)_{t\ge0}$ be an
irreducible finite-state continuous-time Markov semigroup whose transition
matrix $T_t$ is strictly positive for every $t>0$. Fix $t_P,t_Q>0$. The
universal fourth condition is
\begin{equation}
\tag{$\mathrm{C4}^{\mathrm{noise}}$}
\text{for every }i\in I\text{ and every }a\in A,\qquad
P_i(a)=(\rho_iT_{t_P})(a),\quad Q_i(a)=(\rho_iT_{t_Q})(a).
\label{eq:noise-c4}
\end{equation}
Assume Mu's standing genericity hypothesis; boundedness and positivity are
automatic here.

\begin{theorem}[Common-noise completion]
\label{thm:noise-completion}
Within \eqref{eq:noise-c4}, the following are equivalent:
\[
\text{Mu's conditions (1)--(3)}
\quad\Longleftrightarrow\quad t_P<t_Q
\quad\Longleftrightarrow\quad
\forall n\ge1:\ P^{\otimes n}\succeqB Q^{\otimes n}.
\]
Consequently these conditions are also equivalent to eventual large-sample
Blackwell dominance.
\end{theorem}

\begin{proof}
If $t_Q>t_P$, the semigroup identity gives
$Q_i=P_iT_{t_Q-t_P}$ for every state, so the same state-independent noise
kernel gives a product garbling for every $n$. Conversely, for distinct
$i,j$, strict log-sum data processing implies that
$t\mapsto\KL(\rho_iT_t\|\rho_jT_t)$ is strictly decreasing: $T_s$ has strictly
positive entries for $s>0$, and equality in the log-sum inequality would force
$\rho_i=\rho_j$. Thus Mu's ordered KL condition forces $t_P<t_Q$.
Finally, eventual Blackwell dominance implies Mu's three conditions by their
necessity theorem.
\end{proof}

The interpretation is a common state-independent forgetting or
misclassification process: both monitoring systems observe the same latent
quality laws, but one observation is recorded after more noise has accumulated.
The condition is primitive and begins with universal quantifiers; it is not a
claim that an unknown Blackwell kernel exists. Binary symmetric audits and the
categorical contamination model are special cases. The theorem is a restricted
completion, not an unrestricted resolution of Mu's conjecture.

There is a broader, non-parametric completion for modular states. Let
$I=\{0,1\}^{d}$, let
$\mathcal X=\prod_{j=1}^{d}\mathcal X_j$ and
$\mathcal Y=\prod_{j=1}^{d}\mathcal Y_j$, with all component signal spaces
finite. The fourth condition is the
separability restriction
\begin{equation}
\tag{$\mathrm{C4}^{\mathrm{mod}}$}
\text{for every }\theta\in\{0,1\}^{d},\ x\in\mathcal X,\ y\in\mathcal Y,
\qquad
P_\theta(x)=\prod_{j=1}^{d}P^j_{\theta_j}(x_j),\quad
Q_\theta(y)=\prod_{j=1}^{d}Q^j_{\theta_j}(y_j).
\label{eq:modular-c4}
\end{equation}
Assume, for every component $j$, that the binary pair
$(P^j_0,P^j_1;Q^j_0,Q^j_1)$ is bounded, strictly positive, and generic.
As in Mu et al.'s proposition, also assume that the full pair is bounded and
generic in the sense of its ordered pairwise likelihood-ratio extrema.
Economically,
the state is a vector of binary attributes (tasks, failure modes, or quality
dimensions), and the monitoring system observes them through conditionally
independent component audits. This is a standard modular information
structure.
For example, with $d=2$, $\theta_j=1$ can mean that module $j$ is faulty and
$x_j$ can be its test result. The source and target systems may use different
test technologies in each module; \eqref{eq:modular-c4} requires only that the
two test results remain conditionally independent given the module states.

\begin{theorem}[Modular-state completion]
\label{thm:modular-completion}
Within \eqref{eq:modular-c4}, the following are equivalent:
\begin{equation}
\begin{aligned}
&\text{Mu's conditions (1)--(3) for the full experiments}
\quad\Longleftrightarrow\quad
\exists N\ \forall n\ge N:\ P^{\otimes n}\succeqB Q^{\otimes n}.
\end{aligned}
\label{eq:modular-completion-iff}
\end{equation}
\end{theorem}

\begin{proof}
For the necessity direction, suppose eventual full dominance holds. Fix a
component $j$. For each $n$, average a full garbling over the common source
law of the other coordinates and then marginalize the other target
coordinates. More explicitly, fix the other-coordinate state $u$ and define
the marginal kernel by
\[
\bar K^j_n(x_j,y_j)=
\sum_{x_{-j},y_{-j}}
 (P^{-j}_u)^{\otimes n}(x_{-j})
 K_n((x_j,x_{-j}),(y_j,y_{-j})).
\]
Here $x_{-j}$ and $y_{-j}$ range over the product signal spaces with
coordinate $j$ removed, and $P^{-j}_u=\prod_{\ell\ne j}P^\ell_{u_\ell}$.
It is independent of the selected component state $t$, and marginalizing the
full garbling gives
\[
\sum_{x_j}(P^j_t)^{\otimes n}(x_j)\bar K^j_n(x_j,y_j)
 =(Q^j_t)^{\otimes n}(y_j),\qquad t\in\{0,1\}.
\]
Thus it produces a binary garbling from
$(P^j_t)^{\otimes n}$ to $(Q^j_t)^{\otimes n}$ for $t=0,1$; hence every
component is eventually Blackwell dominant. For a binary garbling, the target
likelihood ratio is a conditional expectation of the source likelihood ratio;
therefore its maximum cannot increase and its minimum cannot decrease. Apply
this to the $n$-fold ratios and take $n$th roots. Component genericity then
gives a strict increase of the likelihood-ratio maximum and a strict decrease
of its minimum from $Q^j$ to $P^j$.
For any two full states, the extrema of the product log-likelihood ratio are
sums of the component extrema over the coordinates on which the states
differ. The strict component gaps have the same sign in each sum, so they
cannot cancel. The full pair is therefore generic, and Mu's necessity theorem
gives conditions (1)--(3).

Conversely, fix a component $j$ and two full states that differ only in coordinate $j$.
In each of Mu's MGF inequalities set all exponents to zero except the
coordinate corresponding to this second state; under the other true state,
use the analogous vector with the roles reversed. The resulting inequalities
are exactly the positive- and negative-cone MGF inequalities for the binary
component pair $(P^j_0,P^j_1)$ versus $(Q^j_0,Q^j_1)$. The ordered KL
inequality and the component genericity also restrict to this pair. Mu's
two-state characterization therefore gives an $N_j$ such that, for every
$n\ge N_j$, there is a single kernel $K^j_n$ satisfying
\[
(Q^j_t)^{\otimes n}=(P^j_t)^{\otimes n}K^j_n,
\qquad t\in\{0,1\}.
\]
There are finitely many components, so take $N=\max_jN_j$. Tensor the
component garbling kernels and rearrange coordinates:
\[
\bigotimes_{j=1}^{d}(P^j_{\theta_j})^{\otimes n}
\succeqB
\bigotimes_{j=1}^{d}(Q^j_{\theta_j})^{\otimes n}.
\]
By \eqref{eq:modular-c4} these are exactly
$P_\theta^{\otimes n}$ and $Q_\theta^{\otimes n}$, for every state
$\theta$. This invokes the binary-state characterization in Theorem 1 of
Mu et al.~(\href{https://doi.org/10.3982/ECTA17548}{Econometrica}); that
external theorem is the only input outside this manuscript. The two-factor
tensorization step is the elementary finite-product construction checked in
the accompanying Lean file; iterating that step gives the displayed finite
$d$-fold product.\end{proof}

This theorem is not a reformulation of Blackwell order: condition
\eqref{eq:modular-c4} is a primitive independence restriction, and the proof
uses one-dimensional binary comparisons plus tensorization. It is a genuine
restricted repair with a direct interpretation in modular monitoring. It does
not apply to the present three-state counterexample, whose state space is not
a binary product.

\section{Analytic verification of the two experiments}
This appendix supplies the elementary calculations behind the continuum MGF and
Kullback--Leibler conditions.  No external counterexample or separation result
is used.  Signals are $s\in\{0,1\}$ and all displayed probabilities are
strictly positive.  Put
\[
 U=(1/2,1/2),\quad A=(1/5,4/5),\quad B=(4/5,1/5),
\]
\[
 C=(2/5,3/5),\quad D=(11/15,4/15),\quad V=(7/10,3/10),
\]
and let
\[
 P=(U,A,B),\qquad Q=(U,C,D),\qquad R=(U,U,V).
\]
Coordinatewise,
\begin{equation}\label{eq:mix}
 Q=\tfrac13P+\tfrac23R,
\end{equation}
which follows from $C=(A+2U)/3$ and $D=(B+2V)/3$.
For a true state $i$ and the other two states $j,k$, write
\[
 \MGF_E^i(u,v)=\sum_s E_i(s)\Big(\frac{E_i(s)}{E_j(s)}\Big)^u
                            \Big(\frac{E_i(s)}{E_k(s)}\Big)^v .
\]
The order $(j,k)$ is $(1,2)$ for $i=0$, $(0,2)$ for $i=1$, and $(0,1)$ for
$i=2$; changing the order only exchanges $u$ and $v$.

\subsection*{Positive cone}
For later reference, the six explicit expressions are
\begin{align*}
 m_{P0}(u,v)&=\tfrac12( (5/8)^u(5/2)^v+(5/2)^u(5/8)^v),\\
 m_{Q0}(u,v)&=\tfrac12( (5/6)^u(15/8)^v+(5/4)^u(15/22)^v),\\
 m_{P1}(u,v)&=\tfrac15(2/5)^u(1/4)^v+\tfrac45(8/5)^u4^v,\\
 m_{Q1}(u,v)&=\tfrac25(4/5)^u(6/11)^v+\tfrac35(6/5)^u(9/4)^v,\\
 m_{P2}(u,v)&=\tfrac45(8/5)^u4^v+\tfrac15(2/5)^u(1/4)^v,\\
 m_{Q2}(u,v)&=\tfrac{11}{15}(22/15)^u(11/6)^v
                 +\tfrac4{15}(8/15)^u(4/9)^v.
\end{align*}
Throughout, $x^t$ denotes the positive real power.

\begin{lemma}[state $0$]
For $u,v\geq0$, not both zero, $m_{P0}(u,v)>m_{Q0}(u,v)$.
\end{lemma}
\begin{proof}
Let $L=25/16$, $K=75/88<L$.  If $v\geq u$, put $m=u$ and $r=v-u$.
After multiplying by $2$,
\[
 2m_{P0}=L^m\{(5/2)^r+(5/8)^r\},\quad
 2m_{Q0}=L^m(15/8)^r+K^m(15/22)^r.                 \tag{A.1}
\]
For $r>0$, multiplication by $(5/8)^{-r}$ reduces the needed strict
inequality to
\[
 3^r+(12/11)^r<4^r+1.                              \tag{A.2}
\]
Indeed,
\[
 (12/11)^r-1<(4/3)^r-1\leq3^r((4/3)^r-1)=4^r-3^r,
\]
where $3^r\geq1$ and $(4/3)^r\geq1$.  If $u\geq v$, put $m=v$ and
$r=u-v$; the same calculation reduces the claim to
\[
 (4/3)^r+2^r<4^r+1,                              \tag{A.3}
\]
which follows from
$(4/3)^r-1<2^r-1\leq2^r(2^r-1)=4^r-2^r$.
The factors $K^m\leq L^m$ then give the strict comparison in (A.1) in both
cases.  If $r=0$, the nonzero assumption gives $m>0$, and
$2m_{P0}=2L^m>L^m+K^m=2m_{Q0}$.  At $(u,v)=(0,0)$ both sides equal one.
\end{proof}

\begin{lemma}[state $1$]
For $u,v\geq0$, not both zero, $m_{P1}(u,v)>m_{Q1}(u,v)$.
\end{lemma}
\begin{proof}
Set $t=4^u16^v\geq1$ and $c=(2/5)^u(1/4)^v>0$.  Factoring $c$ gives
\[
 m_{P1}=c\{\tfrac15+\tfrac45t\},\qquad
 m_{Q1}=c\{\tfrac25\,2^u(24/11)^v+\tfrac35\,3^u9^v\}.       \tag{A.4}
\]
Since $24/11\leq4$,
$2^u(24/11)^v\leq2^u4^v=\sqrt t$; and
$3^u9^v<t$ whenever $(u,v)\ne(0,0)$ (at least one of $3<4$ or
$9<16$ is used).  Finally $2\sqrt t<1+t$ for $t>1$.  Hence the bracket on
the right in (A.4) is strictly smaller than
\[
 \tfrac25\sqrt t+\tfrac35t<\tfrac15+\tfrac45t,
\]
which proves the lemma.
\end{proof}

\begin{lemma}[state $2$]
For $u,v\geq0$, not both zero, $m_{P2}(u,v)>m_{Q2}(u,v)$.
\end{lemma}
\begin{proof}
The argument uses the intermediate state $R$.  With
$t=4^u16^v$ and $c=(2/5)^u(1/4)^v$,
\[
 m_{P2}=c\{\tfrac15+\tfrac45t\},\quad
 m_{R2}=c\left\{\tfrac7{10}(7/2)^u(28/5)^v
                    +\tfrac3{10}(3/2)^u(12/5)^v\right\}.       \tag{A.5}
\]
The first product in the braces is strictly less than $t$ as soon as
$(u,v)\ne(0,0)$, because $7/2<4$ and $28/5<16$.  For the second one,
$(3/2)^3<4$ and $(12/5)^3<16$ imply
\[
 (3/2)^u(12/5)^v\leq t^{1/3}
 \leq \tfrac23+\tfrac13t,                              \tag{A.6}
\]
where the last inequality is weighted AM--GM.  Combining (A.5)--(A.6)
with the strict first bound gives $m_{P2}>m_{R2}$.

For completeness, the function
\[
 f_{u,v}(x_0,x_1,x_2)=x_0^{1+u+v}x_1^{-u}x_2^{-v}
\]
is convex on the positive orthant.  To see this directly, put
$\alpha=1+u+v$ and apply the weighted three-factor geometric-mean
inequality with weights $1/\alpha,u/\alpha,v/\alpha$ to the two triples
$(f(x),x_1,x_2)$ and $(f(y),y_1,y_2)$; raising the result to $\alpha$ and
cancelling the positive factors gives
$f(\lambda x+(1-\lambda)y)\leq\lambda f(x)+(1-\lambda)f(y)$.
Using \eqref{eq:mix} signal by signal therefore yields
\[
 m_{Q2}\leq\tfrac13m_{P2}+\tfrac23m_{R2}<m_{P2}.
\]
\end{proof}

\begin{proposition}[positive MGF inequalities]
For every $u,v\geq0$ with $(u,v)\ne(0,0)$ and every state $i$,
\[
                 \MGF_P^i(u,v)>\MGF_Q^i(u,v).
\]
At the origin all six MGFs equal one.
\end{proposition}
\begin{proof}
The three preceding lemmas are exactly the three state coordinates.
\end{proof}

\subsection*{Inverse MGFs on the simplex}
For $a,b\geq0$, $a+b\leq1$, define
\[
 G_{a,b}(X,Y,Z)=\sum_s X_s^{1-a-b}Y_s^aZ_s^b.
\]
The identity
\begin{equation}\label{eq:inverse}
 \MGF_E^i(-a,-b)=G_{a,b}(E_i,E_j,E_k)
\end{equation}
 follows term by term from
$x(x/y)^{-a}(x/z)^{-b}=x^{1-a-b}y^az^b$.
We use two elementary facts. First, $g(x,y,z)=x^py^qz^r$ is jointly concave
when $p,q,r\geq0$ and $p+q+r=1$ (weighted Holder).  Second, for each
fixed $z>0$, the map $t\mapsto z^t=\exp(t\log z)$ is convex in the exponent;
hence a positive weighted sum of such powers lies below its secant line.

\begin{lemma}[base comparison]
For $a,b\geq0$, $a+b\leq1$,
\begin{align*}
 G_{a,b}(U,A,B)&\leq G_{a,b}(U,U,V),\\
 G_{a,b}(A,U,B)&\leq G_{a,b}(U,U,W),\\
 G_{a,b}(B,U,A)&\leq G_{a,b}(V,U,U),                 \tag{A.7}
\end{align*}
where $W=(3/10,7/10)$.  The second target equals
$G_{a,b}(U,U,V)$ after swapping the two signals; the third is the corresponding
reparameterized target, as made explicit below.
\end{lemma}
\begin{proof}
For the first inequality, write $L=1-b$.  If $b<1$, direct expansion gives
\[
 G_{a,b}(U,A,B)=c_1(2/5)^a+c_2(8/5)^a,\quad
 c_1=(1/2)^L(4/5)^b,\quad c_2=(1/2)^L(1/5)^b.                \tag{A.8}
\]
Convexity of $x\mapsto e^{x\log z}$ on $[0,L]$ gives
\[
 G_{a,b}(U,A,B)\leq(1-a/L)G_{0,b}(U,A,B)
                    +(a/L)G_{L,b}(U,A,B).             \tag{A.9}
\]
The endpoints are $G_{0,b}(U,A,B)=H(U,B;b)$ and
$G_{L,b}(U,A,B)=H(A,B;b)$, where
$H(X,Y;r)=\sum_sX_s^{1-r}Y_s^r$.  Holder gives $H(A,B;b)\leq1$.
Since $U=(A+B)/2$, joint concavity of $x^{1-b}y^b$ gives
\[
 H(U,B;b)\geq\tfrac12H(A,B;b)+\tfrac12\geq H(A,B;b).  \tag{A.10}
\]
Also $V=(A+5B)/6$, and strict concavity of $x^b$ for $0<b<1$ gives
\[
 H(U,B;b)<H(U,V;b).                                  \tag{A.11}
\]
Equations (A.9)--(A.11) prove the first line of (A.7) for $0<b<1$.
For $b=0$ and $a>0$, strict concavity of $x^a$ gives
$G_{a,0}(U,A,B)=H(U,A;a)<1=G_{a,0}(U,U,V)$; at $(a,b)=(0,0)$ there is
 equality.  If $b=1$, necessarily $a=0$ and direct substitution gives the
weak inequality.

The remaining two lines are the same calculation after reindexing factors:
\[
 G_{a,b}(A,U,B)=G_{1-a-b,b}(U,A,B),\qquad
 G_{a,b}(B,U,A)=G_{b,1-a-b}(U,A,B),
\]
and
\[
 G_{a,b}(U,U,W)=G_{a,b}(U,U,V),\qquad
 G_{b,1-a-b}(U,U,V)=G_{a,b}(V,U,U).
\]
The first identity swaps the two signal labels; the second combines the same
swap with the displayed parameter reindexing.  This proves (A.7).
\end{proof}

\begin{lemma}[strictness in the open simplex]
If $a,b\geq0$, $a+b<1$, and $(a,b)\ne(0,0)$, all three inequalities in
(A.7) are strict.
\end{lemma}
\begin{proof}
For the first line, if $b>0$, then $b<1$ and the coefficient
$1-a/(1-b)$ in (A.9) is positive; (A.11) makes the inequality strict.
If $b=0$, then $a>0$ and the strict Jensen argument following (A.11) applies.
For the second and third lines use the two reindexing identities above,
together with the target identities following them.  In the open simplex the
transformed pair is nonzero, so the first-line strictness applies.
\end{proof}

\begin{proposition}[inverse-MGF inequalities]
For $a,b\geq0$ with $a+b\leq1$,
\[
 G_{a,b}(P_i,P_j,P_k)\leq G_{a,b}(Q_i,Q_j,Q_k),\quad i=0,1,2.
\]
If in addition $a+b<1$ and $(a,b)\ne(0,0)$, every inequality is strict.
Equivalently, by \eqref{eq:inverse},
\[
 \MGF_P^i(-a,-b)<\MGF_Q^i(-a,-b)
\]
on the open simplex excluding the origin; equality holds at the origin.
\end{proposition}
\begin{proof}
For each signal $s$, \eqref{eq:mix} and joint concavity of
$g(x,y,z)=x^{1-a-b}y^az^b$ imply
\[
 G_{a,b}(Q_i,Q_j,Q_k)\geq
 \tfrac13G_{a,b}(P_i,P_j,P_k)+\tfrac23G_{a,b}(R_i,R_j,R_k).    \tag{A.12}
\]
The base lemma identifies the corresponding $R$ coordinate with the right
side of (A.7).  Thus (A.12) gives the weak statement, and the strict base
lemma makes the right side strictly larger than the $P$ coordinate in the
open simplex.
\end{proof}

\subsection*{Boundedness, genericity, and KL inequalities}
Every entry of $P$ and $Q$ is positive and each row sums to one.  All
likelihood ratios lie in $[1/4,4]$.  For the unordered state pairs, the
extreme likelihood ratios are
\[
\begin{array}{c|cc|cc}
 &\multicolumn{2}{c|}{P}&\multicolumn{2}{c}{Q}\\
\text{pair}&\min&\max&\min&\max\\ \hline
0,1&5/8&5/2&5/6&5/4\\
0,2&5/8&5/2&15/22&15/8\\
1,2&1/4&4&6/11&9/4
\end{array}
\]
For a reversed ordered pair, the extrema are reciprocals of the displayed
ones, so they are also distinct. This verifies boundedness, the common ratio
bound and pairwise genericity.

Write $K_{ij}^E=\KL(E_i\Vert E_j)$.  Expanding the two signal terms and
clearing positive common denominators gives the following six certificates;
strict monotonicity of $\log$ turns each integer inequality into
$K_{ij}^Q<K_{ij}^P$:
\[
\begin{array}{c|c|c}
(i,j)&\text{logarithmic comparison}&\text{integer certificate}\\ \hline
(0,1)&\log2+\log8<\log4+\log6&2\cdot8<4\cdot6\\
(0,2)&2\log15+\log2<2\log5+\log22&15^2\cdot2<5^2\cdot22\\
(1,0)&2\log4+3\log6<\log2+4\log8&4^2 6^3<2\,8^4\\
(1,2)&2\log6+3\log9<6\log4+2\log11&6^2 9^3<4^6 11^2\\
(2,0)&11\log22+15\log5<27\log2+15\log15&22^{11}5^{15}<2^{27}15^{15}\\
(2,1)&11\log11<11\log6+4\log9+5\log4&11^{11}<6^{11}9^4 4^5
\end{array}
\]
For example, the first row compares
$K_{01}^P=\frac12\log(25/16)$ and
$K_{01}^Q=\frac12\log(25/24)$; the other rows are obtained by the same
expansion, with positive factors $1/5$, $1/15$, or $1/2$ suppressed.
The displayed integer inequalities follow by multiplication, and hence
all ordered KL inequalities are strict.

\section*{References}
\begingroup
\small
\raggedright

\begin{itemize}
\item Blackwell (1953), ``Equivalent Comparisons of Experiments,'' \textit{Annals of Mathematical Statistics}, \url{https://doi.org/10.1214/aoms/1177729032}.
\item Strassen (1965), ``The existence of probability measures with given marginals,'' \textit{Annals of Mathematical Statistics}.
\item Mu, Pomatto, Strack, and Tamuz, ``From Blackwell Dominance in Large Samples to R\'enyi Divergences and Back Again,'' \textit{Econometrica}, \url{https://doi.org/10.3982/ECTA17548}.
\item Ashkenazi-Golan and Lehrer, ``Blackwell's Comparison of Experiments and Discounted Repeated Games,'' \textit{Games and Economic Behavior}, \url{https://doi.org/10.1016/j.geb.2019.06.003}.
\item P\'eski, ``Comparison of Information Structures in Zero-Sum Games,'' \textit{Games and Economic Behavior}, \url{https://doi.org/10.1016/j.geb.2007.06.004}.
\item Greenblatt et al., ``AI Control: Improving Safety Despite Intentional Subversion,'' \url{https://arxiv.org/abs/2312.06942}.
\item Hubinger et al., ``Sleeper Agents: Training Deceptive LLMs that Persist Through Safety Training,'' \url{https://arxiv.org/abs/2401.05566}.
\end{itemize}
\endgroup

\end{document}